\documentclass[journal]{IEEEtran}      
\usepackage{graphicx}
\usepackage{graphics}
\graphicspath{ {} }
\usepackage{amsmath}
\usepackage{amsthm}
\usepackage[english]{babel}
\usepackage{amssymb}
\usepackage[usenames, dvipsnames]{color}

\newtheorem{proposition}{Proposition}
\newtheorem{theorem}{Theorem}
\newtheorem{remark}{Remark}

\title{Synchronization with Guaranteed Clock Continuity using Pulse-Coupled Oscillators}
\author{Timothy Anglea\(^{1}\) (\textit{IEEE Student Member}) \ \ Yongqiang Wang\(^{1}\) 
	\thanks{ $^{1}$Timothy Anglea and Yongqiang Wang are with the department of Electrical \& Computer Engineering, Clemson University, Clemson, SC 29634, USA {\{tbangle,yongqiw\}@clemson.edu}}
	(\textit{IEEE Senior Member})}
\date{\today}

\begin{document}		

\maketitle
\begin{abstract}
	Clock synchronization is a widely discussed topic in the engineering literature. Ensuring that individual clocks are closely aligned is important in network systems, since the correct timing of various events in a network is usually necessary for proper system implementation. However, many existing clock synchronization algorithms update clock values abruptly, resulting in discontinuous clocks which have been shown to lead to undesirable behavior. In this paper, we propose using the pulse-coupled oscillator model to guarantee clock continuity, demonstrating two general methods for achieving continuous phase evolution in any pulse-coupled oscillator network. We provide rigorous mathematical proof that the pulse-coupled oscillator algorithm is able to converge to the synchronized state when the phase continuity methods are applied. We provide simulation results supporting these proofs. We further investigate the convergence behavior of other pulse-coupled oscillator synchronization algorithms using the proposed methods.
	
	Index - clock continuity, synchronization, pulse-coupled oscillators, phase jumps
\end{abstract}

\section{Introduction} \label{sec:intro}

Ensuring clock synchronization in a distributed system is a very important and well-studied topic in the fields of electrical engineering and computer science. However, many synchronization algorithms require instantaneous clock value adjustment \cite{ATSClockSync2011}, which leads to discontinuous clocks. These discontinuous clocks are undesirable. If time jumps forward, then there is the potential that a scheduled event will never happen, and if time jumps backwards, then one process may be implemented twice, as shown in Fig. \ref{fig:ClockJump}. It is desirable that the time value evolve continuously while synchronizing with the other sensors within the network \cite{Lamport85, OptimalClockSync87}.

\begin{figure}[t] 
	\includegraphics[width=\columnwidth]{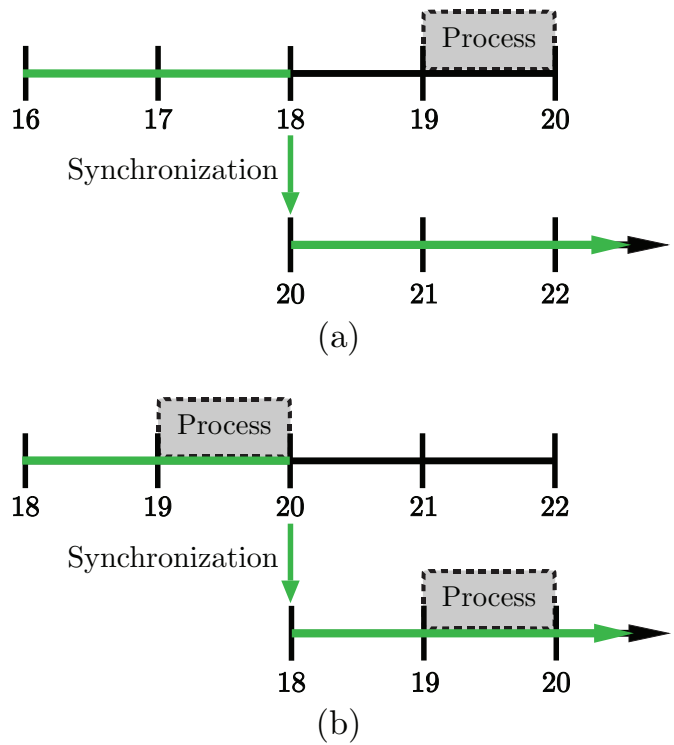}
	\centering
	\caption{Illustration of disadvantages of discontinuous clock synchronization. (a) Synchronization occurs after a jump forward of two time units, and the scheduled process is not executed. (b) Synchronization occurs after a jump backward of two time units, and the scheduled process is executed twice.}
	\label{fig:ClockJump}
\end{figure}


Many clock synchronization algorithms use packet-based communication to share local information \cite{ContinuousClockMock2000, TimingConstraint2001, ClockSyncSurvey2005} and can achieve synchronization with a continuous clock \cite{Lamport85, OptimalClockSync87}. However, as noted in \cite{TimingConstraint2001}, such approaches may require constant adjustment to the clock rate, leading to significant runtime overhead. In this paper, we will show that the pulse-coupled oscillator (PCO) model can be used to achieve guaranteed clock continuity.

The PCO model was first introduced by Peskin in 1975 \cite{peskin:75}. He used pulse-coupled oscillators (PCOs) to model the synchronization of pacemaker cells in the heart. Mirollo and Strogatz later improved the model, providing a rigorous mathematical formulation \cite{mirollo:90}. Communication latency, packet loss, signal corruption, and energy consumption are all minimized due to the simplicity of the communication between oscillators in the network (i.e., single pulses).

Recently, the PCO model has been widely used to synchronize clocks in wireless sensor networks. However, all of the existing PCO synchronization algorithms achieve synchronization through abrupt jumps in the oscillator phase \cite{Werner-Allen:2005:FSN:1098918.1098934, ScalableSync2005, OptimalPRCSync, EnergyEfficientSync2012, Dorfler20141539, NishimuraSync2015, Felipe_TAC15, Felipe_Automatica, Proskurnikov2015, PulseSS2016, KlinglmayrSync2016, Brandner2016}, which, as indicated earlier, may lead to undesired behavior. In this paper, we propose a new synchronization scheme for PCOs, which can guarantee continuity in phase evolution for any jump-based PCO algorithm. We will show that the convergence properties of our previously proposed jump-based synchronization algorithms in \cite{OptimalPRCSync, EnergyEfficientSync2012} are maintained, even when the phase adjustment rule is modified to be continuous. To do so, we will prove that the new mechanism guaranteeing phase continuity amounts to reducing the coupling strength in the conventional PCO model. To our knowledge, this is also the first paper to consider a time-varying coupling strength in a PCO network.

In Section \ref{sec:prelim}, we will define our terms and notation for a PCO network. We will then introduce two methods for maintaining the continuity of the phase of an oscillator in Section \ref{sec:continuity}. In Section \ref{sec:PCOanalysis}, we will analyze the behavior of the network under continuous phase evolution, and show that the behavior can be modeled as a time-varying coupling strength. In Section \ref{sec:SyncProofs}, we will analyze the convergence properties of PCOs under time-varying coupling strengths, and show that a PCO network will synchronize under guaranteed phase continuity. We will then present simulations to further illustrate phase continuity in Section \ref{sec:results}. We will conclude with final remarks and future work in Section \ref{sec:conclusion}.

\section{Pulse Coupled Oscillators Preliminaries} \label{sec:prelim} 

Let us consider a network of \(N\) identical pulse-coupled oscillators. Let \(\theta_i \in [0,1)\) be the associated phase of oscillator \(i \in \mathcal{V} = \{1,2,\cdots,N\}\). For our analysis, each oscillator has an identical fundamental frequency, \(\omega_0\), and evolves at that rate on the interval \([0,1)\).

As the network evolves, each oscillator increases its phase at rate \(\omega_0\). When an oscillator reaches the threshold value of \(1\), it fires a pulse and resets its phase to zero. Any connected oscillators then receive that pulse, being notified of the firing instance of an oscillator in the network. Receiving a pulse will cause an oscillator to change its phase in accordance with the control algorithm. Let us denote the amount that oscillator \(i\) determines to adjust its phase, given a PCO control algorithm, at a firing instance at time \(t\) as \(\psi_i\), where
\begin{equation} \label{eq:PhaseChange}
\psi_i = \alpha \phi_i = \lim_{\tau \downarrow 0} \big(\theta_i(t + \tau)\big) - \theta_i(t) = \theta_i(t^{+}) - \theta_i(t)
\end{equation}
where \(\alpha\) is the coupling strength of the network, and \(\phi_i\) is the phase change amount determined by the PCO control algorithm. The range of the coupling strength allowed for the network is determined by the algorithm, but typical values of coupling strength are within the interval \((0,1]\).

\section{Pulse Coupled Oscillators with Continuous Phase Evolution} \label{sec:continuity} 
To our knowledge, all existing literature regarding PCO networks has the phase value of each oscillator jump discontinuously at firing instances. In this section, we will enhance the PCO model such that the phase value must evolve continuously at all times (except when it resets its phase to zero when it reaches the threshold). Thus, under this restriction, when oscillator \(i\) receives a pulse, it must increase or decrease its individual rate of evolution, \(\omega_i\), for a certain amount of time, \(\tau_i\), in order to achieve the required phase adjustment, \(\psi_i\), while ensuring continuity. If oscillator \(i\) receives another pulse before the time needed to achieve \(\psi_i\) is completed, then the oscillator will use its current phase \(\theta_i\) to redetermine a new \(\psi_i\), and thus determine new values for \(\omega_i\) and \(\tau_i\). 

We propose two methods for adjusting the phase of an oscillator in a continuous fashion: either the frequency at which the oscillator evolves in response to receiving a pulse is kept constant, or the time for which the oscillator evolves is kept constant.

\subsection{Constant Frequency Method} 
In the constant frequency method, the frequency of oscillator \(i\) is increased or decreased by a set amount \(\omega_a\) for an adjustable duration of time \(\tau_i\). The amount of time the oscillator spends at this new frequency is dependent on the phase amount \(\psi_i\) that it needs to adjust:
\begin{equation} \label{eq:FreqMethodTime}
\tau_i = \frac{|\psi_i|}{\omega_a}
\end{equation}

Thus, once an amount of phase adjustment \(\psi_i\) is determined, the oscillator will increase its frequency to \(\omega_i = \omega_0 + \omega_a\) if \(\psi_i\) is positive, or decrease its frequency to \(\omega_i = \omega_0 - \omega_a\) if \(\psi_i\) is negative, for time \(\tau_i\) determined in (\ref{eq:FreqMethodTime}). Once the time \(\tau_i\) has elapsed, the oscillator returns to its fundamental frequency \(\omega_0\). If \(\psi_i\) is zero, then the oscillator remains at its fundamental frequency, \(\omega_0\), and evolves until the next firing instance.

\begin{remark} \label{r:ConstFreqRemark}
	In the constant frequency method, it is possible to have different amounts of frequency change when increasing or decreasing the oscillator's frequency, i.e., \(\omega_{a}^{+}\) and \(\omega_{a}^{-}\) respectively. In this paper, we will focus on using the same amount of frequency change \(\omega_a = \omega_{a}^{+} = \omega_{a}^{-}\) for both increasing and decreasing the frequency of the oscillator.
\end{remark}

\subsection{Constant Time Method} 
In the constant time method, oscillator \(i\) spends a fixed amount of time \(\tau\) at an adjustable frequency \(\omega_i\). The new frequency at which the oscillator evolves is dependent on the phase amount \(\psi_i\) the oscillator needs to adjust:
\begin{equation} \label{eq:TimeMethodFreq}
\omega_i = \omega_0 + \omega_a = \omega_0 + \frac{\psi_i}{\tau}
\end{equation}

Note that \(\omega_a\) can be positive or negative. Thus, once an amount of phase adjustment \(\psi_i\) is determined, the oscillator will update its frequency to \(\omega_i = \omega_0 + \omega_a\) for time \(\tau\). Once the fixed amount of time \(\tau\) has elapsed, the oscillator again returns to its fundamental frequency \(\omega_0\).

\begin{remark} \label{r:PriorMethods}
	The constant time method described above is a general case of how other algorithms ensure clock continuity. In packet-based synchronization algorithms, the value of \(\tau\) is set to be the length of the communication period \cite{OptimalClockSync87,ContinuousClockMock2000}.
\end{remark}

\begin{remark} \label{r:MethodGeneralization}
	Phase jumps, as is standard in the literature, can be seen as a specific case of either of the above phase continuity methods. These cases can be obtained by either taking the limit as \(\omega_a\) goes to infinity in the constant frequency method, or by taking the limit as \(\tau\) goes to zero in the constant time method.
\end{remark}

\begin{remark} \label{r:NegativeFrequency}
	Depending on the parameters chosen in each of the above phase continuity methods, the oscillators may evolve backward in phase (i.e., \(\omega_i < 0\)). Negative frequencies are acceptable in the analysis used in this paper. However, parameters can be chosen to ensure that \(|\omega_a| < 1\) such that the oscillator frequency remains strictly positive.
\end{remark}

\begin{remark} \label{r:ComparisonOfMethods}
	Both methods above achieve the same basic result of having the phase evolve continuously. However, each has their desirable characteristics. The constant frequency method only requires the oscillators to evolve at a countable set of frequencies. The constant time method, however, ensures that the phase adjustment occurs in a set amount of time. The method used should take into consideration the specific application of the PCO network.
\end{remark}

\section{Oscillator Analysis} \label{sec:PCOanalysis} 
We now rigorously analyze the dynamics of oscillators maintaining continuous phase evolution according to the proposed methods above, rather than using phase jumps. Once a method is chosen, all that is required is to take the amount of phase adjustment for oscillator \(i\), i.e. \(\psi_i\), and determine the necessary amount of time, \(\tau_i\), and frequency increase/decrease \(\omega_a\). We then let the oscillator evolve at the new frequency for the required amount of time before it returns to its fundamental frequency, \(\omega _0\).

\subsection{Single Oscillator Behavior}
Let us analyze the behavior of a single oscillator. Once oscillator \(i\) has calculated the necessary change in frequency, \(\omega_i\), to achieve the phase adjustment, \(\psi_i\), in time \(\tau_i\), two possibilities can follow: 1) the oscillator receives no pulses before time \(\tau_i\) that cause a phase adjustment, and 2) the oscillator receives a pulse before time \(\tau_i\) that causes a phase adjustment. If the length of the time interval between the current and previous oscillator firing instances is given as \(t_0\), we can divide these two cases mathematically as 1) \(t_0 \geq \tau_i\), 2) \(t_0 < \tau_i\).

\begin{enumerate}
	\item In the first case, oscillator \(i\) finishes adjusting its phase by \(\psi_i\), and returns to evolving at the fundamental frequency \(\omega_0\). The same effective change in phase has been achieved as if the oscillator had jumped in phase by \(\psi_i\) and evolved normally for a time of length \(t_0\). Thus, no effective change in the phase update rule occurs compared with the conventional instantaneous jump-based PCO model.
	
	\item In the second case, the oscillator has not yet achieved its desired amount of phase change. Rather than having adjusted the whole amount \(\psi_i\), it has adjusted only a portion of that amount, \(\frac{t_0}{\tau_i}\psi_i\), in the time interval between received pulses. It then will use its current phase at the time when the new pulse is received to redetermine new values for \(\psi_i\), \(\tau_i\), and \(\omega_i\). This phase evolution is equivalent to having the oscillator jump in phase by \(\frac{t_0}{\tau_i}\psi_i\), and then evolve normally for a time of length \(t_0\). This fractional amount of the desired phase change can be seen as a reduction of the coupling strength, \(\alpha\), of oscillator \(i\) by the ratio \(\frac{t_0}{\tau_i}\). From (\ref{eq:PhaseChange}), we can write an expression for the effective coupling strength, \(\alpha_{e_i}\), of the oscillator:
	\begin{equation} \label{eq:EffectiveAlpha}
	\alpha_{e_i} = \frac{t_0}{\tau_i} \alpha
	\end{equation}
\end{enumerate}

This second case leads us to analyze a standard PCO network under the condition of a time-varying coupling strength.

\subsection{Time-Varying Coupling Strength}

Let us consider a PCO network under a control algorithm that allows jumps in the phase variable, \(\theta\), but has a coupling strength, \(\alpha\), that varies with time.
\begin{proposition} \label{prop:FiringInstanceAlpha}
	The evolution of an oscillator in a PCO network is dependent only upon the value of the coupling strength, \(\alpha\), at firing instances.
\end{proposition}
\begin{proof}
	The proof for this proposition is straightforward. An oscillator only determines the amount that it needs to jump when it receives a pulse. Thus, the value of the coupling strength is only used at firing instances. Any values the coupling strength takes between firing instances is unused and thus independent of the behavior of the network. Furthermore, if the oscillator receives a pulse, but does not jump (or jumps an amount of zero), then the coupling strength is again independent to the behavior of the network.
\end{proof}

\begin{remark} \label{r:ConvergenceTime}
	Typically, a larger network coupling strength will cause the network to converge more quickly. When considering phase continuity, the effective coupling strength will be less than or equal to the actual coupling strength of the network, implying that phase continuity will cause the network to converge more slowly in general. The lengthening of convergence time is dependent on the parameters chosen for each method.
\end{remark}

\begin{remark} \label{r:AlphaNotAdjusted}
	Note that the phase continuity methods described above do not actively modify the coupling strength between time instances. Only the apparent behavior of the network is being modeled as a reduced coupling strength, \(\alpha_e\). The actual coupling strength, \(\alpha\), remains unchanged throughout the entire evolution of the network.
\end{remark}

\section{Synchronization Analysis} \label{sec:SyncProofs}

We will now use the PCO synchronization strategy given in \cite{EnergyEfficientSync2012} with phase jumps to determine if the convergence properties of the algorithm still hold under a time-varying coupling strength, and thus under the newly proposed phase continuity methods. The algorithm in \cite{EnergyEfficientSync2012} uses a delay-advance phase response curve (PRC) to describe the phase update at firing instances.

Consider a PCO network with \(N\) oscillators in a general, (strongly) connected graph. Without loss of generality, we can order the oscillators in the network according to their phase, such that oscillator \(1\) has the smallest phase and oscillator \(N\) has the largest phase (i.e., \(0 \leq \theta_1 \leq \ldots \leq \theta_N < 1\)). When an oscillator receives a pulse, it updates its phase variable according to the phase response curve, or function, \(Q\), as shown in Fig. \ref{fig:PRC1}.

\begin{figure}[t] 
	\includegraphics[width=\columnwidth]{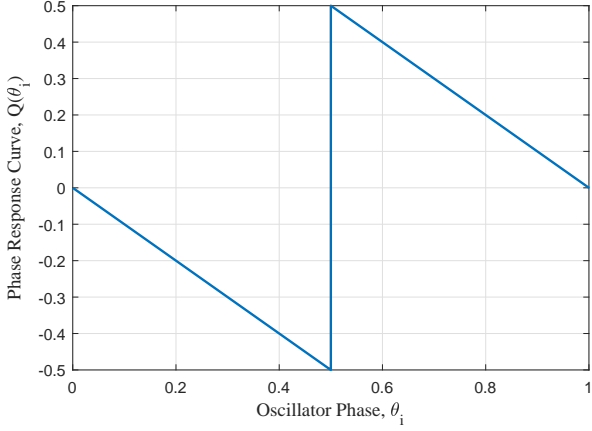}
	\centering
	\caption{Phase Response Curve (PRC) for PRC synchronization given in (\ref{eq:PRCSync})}
	\label{fig:PRC1}
\end{figure}

\begin{equation} \label{eq:PRCSync}
Q(\theta_i) = \left\{ \begin{array}{ll}
			-\theta_i & \mbox{if \(0 \leq \theta_i \leq \frac{1}{2}\)}\\
			(1 - \theta_i) & \mbox{if \(\frac{1}{2} < \theta_i < 1\)} \end{array} \right.
\end{equation}

Note that the phase update is independent of the number and relative positions of other oscillators in the network. Thus, the phase of oscillator \(i\) after a firing instance can be described as
\begin{equation} \label{eq:SyncUpdate}
\theta_i(t^+) = \theta_i(t) + \alpha Q(\theta_i(t))
\end{equation}
where \(\alpha\) is the coupling strength of the network, and, from (\ref{eq:PhaseChange}), we have \(Q(\theta_i(t)) = \phi_i\).

A refractory period of length \(D\) can be included in the phase response curve, as shown in Fig. \ref{fig:PRCRefractory1}. An oscillator does not respond to incoming pulses if its phase is within the region \([0,D)\), and continues to freely evolve. Such a refractory period can improve energy-efficiency and robustness to communication latency \cite{EnergyEfficientSync2012}.

\begin{figure}[t] 
	\includegraphics[width=\columnwidth]{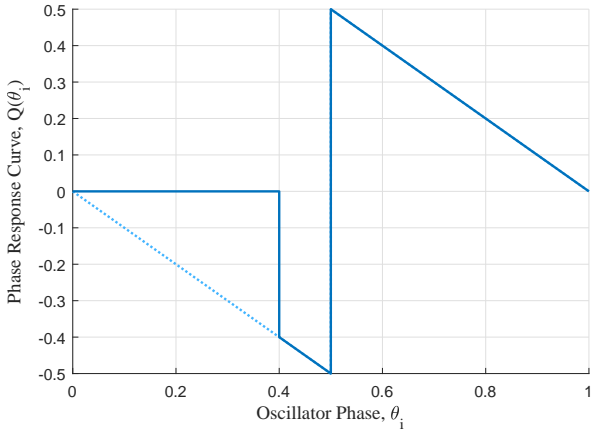}
	\centering
	\caption{Phase Response Curve (PRC) for PRC synchronization given in (\ref{eq:PRCSync}) with an example refractory period \(D = 0.4\).}
	\label{fig:PRCRefractory1}
\end{figure}

Let us refer to an arc as a connected subset of the interval \([0,1)\). We can define the following set of functions, \(v_{i,i+1}\), for all \(i \in \mathcal{V}\):
\begin{equation} \label{eq:PhaseDiffArcs}
v_{i,i+1}(\theta) = \left\{ \begin{array}{ll}
			\theta_{i+1} - \theta_{i} & \mbox{if \(\theta_{i+1} > \theta_{i}\)}\\
			1 - (\theta_{i+1} - \theta_{i}) & \mbox{if \(\theta_{i} > \theta_{i+1}\)} \end{array} \right.
\end{equation}
where oscillator \(N+1\) maps to oscillator \(1\). Notice that these functions do not change between firing instances.

We say that the containing arc of the oscillators is the smallest arc that contains all of the phases in the network. The length of this arc, \(\Lambda\), is given mathematically as

\begin{equation} \label{eq:ContainingArc}
\Lambda = 1 - \max_{i \in \mathcal{V}} \{v_{i,i+1}(\theta)\}
\end{equation}

As the network synchronizes, the length of the containing arc decreases and converges to zero. This value has been shown in \cite{EnergyEfficientSync2012} to decrease after all oscillators have fired once with a constant coupling strength if it is initially less than some \(\bar{\Lambda} \in (0,\frac{1}{2}]\). We will next show that this quantity also decreases monotonically under a bounded, time-varying coupling strength.

\begin{theorem} \label{thr:PRCSyncTheorem}
	Consider pulse-coupled oscillators with a refractory period \(D\) in the phase response curve in (\ref{eq:PRCSync}), as in Fig. \ref{fig:PRCRefractory1}. A strongly connected network of such PCOs using phase jumps will synchronize with oscillators having independent, time-varying, and bounded  \(\alpha \in (0,1]\) if the containing arc of the oscillators is less than some \(\bar{\Lambda} \in (0,\frac{1}{2}]\), and if the refractory period \(D\) is not greater than \(1 - \bar{\Lambda}\).
\end{theorem}
\begin{proof}
	Let us consider a PCO network where the initial phases are within some containing arc \(\Lambda < \bar{\Lambda}\). Without loss of generality, let us assume that oscillator \(i\) has the largest initial phase, \(\theta_{\max}\) at time \(t = 0\), oscillator \(j\) has the smallest initial phase such that \(\theta_j = \theta_{\max} - \Lambda\), and all other oscillator phases reside between oscillator \(i\) and \(j\).
	
	Since oscillator \(i\) has the largest phase, its phase evolves to \(1\) without perturbation and it reaches the threshold at \(t = \frac{1 - \theta_{\max}}{\omega_0}\). At this firing instance, all of the other oscillators have phases between \(1 - \Lambda\) (which is larger than \(\frac{1}{2}\)) and \(1\). In the following time interval of length \(\frac{\Lambda}{\omega_0}\), every oscillator will fire once. Since the network is strongly connected, oscillator \(j\) receives at least one pulse during its phase evolution from \(1 - \Lambda\) to \(1\), and its phase is increased. (The value of phase response curve is positive in the interval \((\frac{1}{2},1)\).) We denote the phase increase as \(\phi_j\), which is strictly positive and dependent on the time-varying coupling strength, \(\alpha \in (0,1]\), and the phase response curve, and hence is time-dependent. Given that the initial phase difference is \(\Lambda\), and that the phase of oscillator \(j\) is increased by \(\phi_j\), then the containing arc of the network decreases by at least \(\phi_j\), as oscillator \(i\) may have decreased its phase due to the pulse received while in the interval \([D,\Lambda)\), if \(D < \Lambda\) holds. (The value of the phase response curve is negative in the interval \((0,\frac{1}{2})\)). The network then continues on to the next cycle, and the above analysis repeats.
	
	Therefore, since the containing arc of the network decreases with every cycle, and cannot be negative, then the containing arc converges to zero, and the network synchronizes.	
\end{proof}

\begin{remark} \label{r:SyncGeneralAlpha}
	The coupling strengths of the oscillators can vary independently from each other, and synchronization will still occur, as long as the coupling is bounded (i.e., \(\alpha \in (0,1]\)).
\end{remark}

Theorem \ref{thr:PRCSyncTheorem} proves that a PCO network can synchronize for any potentially time-varying coupling strength, \(\alpha\). As shown in Sec. \ref{sec:continuity}, phase continuity can be modeled as a reduction in the coupling strength of an oscillator as in (\ref{eq:EffectiveAlpha}). Since the bound of this effective coupling strength is \((0,\alpha]\), the effective coupling strength will also be inside the bound \((0,1]\). Thus, we easily have the following theorem. 
 
\begin{theorem} \label{thr:ContinuousPRCSyncTheorem}
	Consider pulse-coupled oscillators with a refractory period \(D\) in the phase response curve in (\ref{eq:PRCSync}), as in Fig. \ref{fig:PRCRefractory1}. A strongly connected network of such PCOs under one of the phase continuity methods in Sec. \ref{sec:continuity} will synchronize with \(\alpha \in (0,1]\) if the containing arc of the oscillators is less than some \(\bar{\Lambda} \in (0,\frac{1}{2}]\), and if the refractory period \(D\) is not greater than \(1 - \bar{\Lambda}\).
\end{theorem}
\begin{proof}
	This proof follows easily from Theorem \ref{thr:PRCSyncTheorem}. Phase continuity will result in oscillator coupling strengths being independent, time-varying, and bounded within the interval \((0,\alpha]\), as shown in the analysis of section \ref{sec:continuity}. Thus the conditions for Theorem \ref{thr:PRCSyncTheorem} are met, and the network will converge to the state of synchronization.
\end{proof}

\section{Simulations and Results} \label{sec:results}
\begin{figure}[t] 
	\includegraphics[width=\columnwidth]{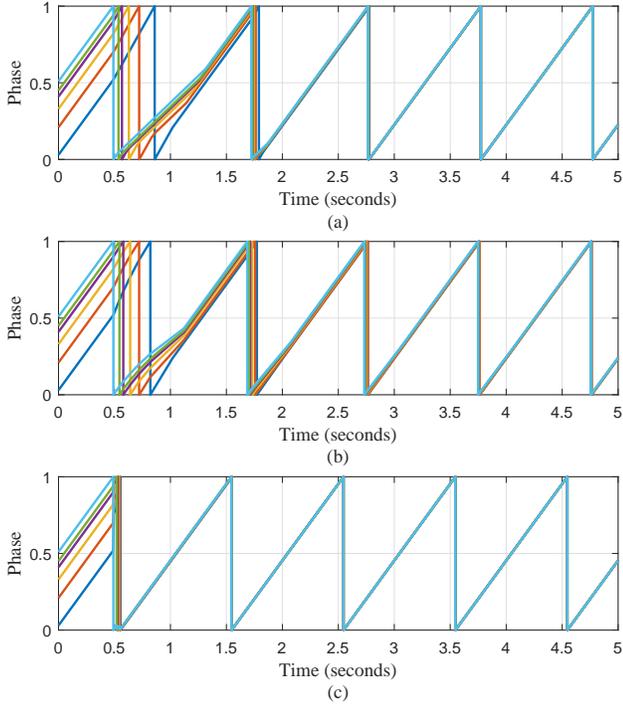}
	\centering
	\caption{Phase evolution of the PRC synchronization algorithm in \cite{EnergyEfficientSync2012} for \(N = 6\) oscillators in an all-to-all topology, \(\alpha = 0.5\), no refractory period, and random initial conditions in a containing arc \(\Lambda < \frac{1}{2}\). (a) Continuous phase evolution under the constant frequency method, with \(\omega_a = 0.3\omega_0\). (b) Continuous phase evolution under the constant time method, with \(\tau = 0.3\) seconds. (c) Phase jumps, for comparison to the phase continuity methods.}
	\label{fig:PRCSyncPhaseEvolution1}
\end{figure}

\begin{figure}[t] 
	\includegraphics[width=\columnwidth]{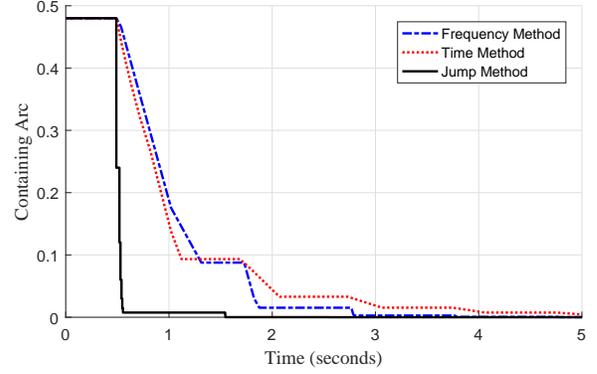}
	\centering
	\caption{Containing arcs, \(\Lambda\) as a function of time for the networks in Fig. \ref{fig:PRCSyncPhaseEvolution1}. The convergence speed of the containing arcs under the constant frequency method, with \(\omega_a = 0.3\omega_0\), and the constant time method, with \(\tau = 0.3\) seconds, is reduced compared with the phase jump case.}
	\label{fig:PRCSyncArcs1}
\end{figure}

We now simulate PCO synchronization algorithms using the phase continuity methods as discussed above. All simulations are performed in MATLAB, with oscillators evolving over the interval \([0,1)\), with fundamental frequency \(\omega_0 = 1\), and period of one second.

\begin{figure}[t] 
	\includegraphics[width=\columnwidth]{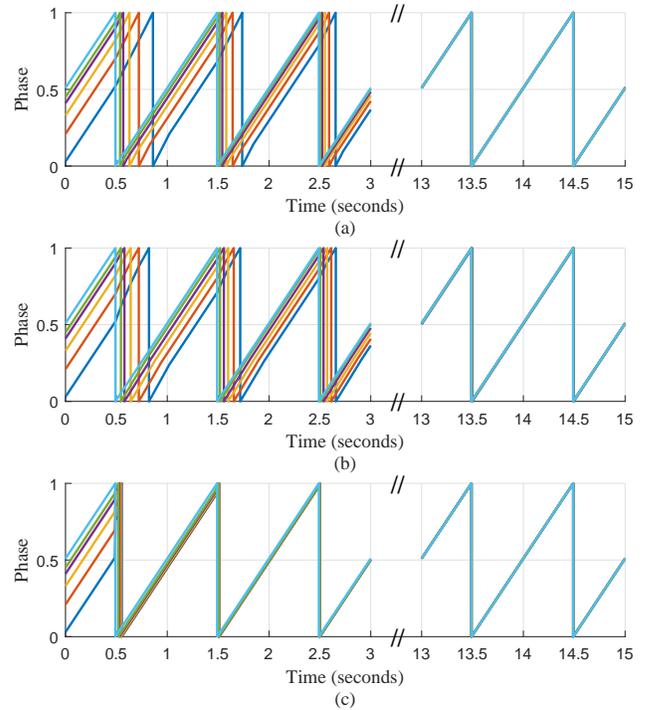}
	\centering
	\caption{Phase evolution of the Energy-Efficient Synchronization algorithm in \cite{EnergyEfficientSync2012} for \(N = 6\) oscillators in an all-to-all topology, \(\alpha = 0.5\), refractory period \(D = 0.5\), and random initial conditions in a containing arc \(\Lambda < \frac{1}{2}\). (a) Continuous phase evolution under the constant frequency method, with \(\omega_a = 0.3\omega_0\). (b) Continuous phase evolution under the constant time method, with \(\tau = 0.3\) seconds. (c) Phase jumps, for comparison to the phase continuity methods.}
	\label{fig:PRCSyncPhaseEvolution2}
\end{figure}

\begin{figure}[t] 
	\includegraphics[width=\columnwidth]{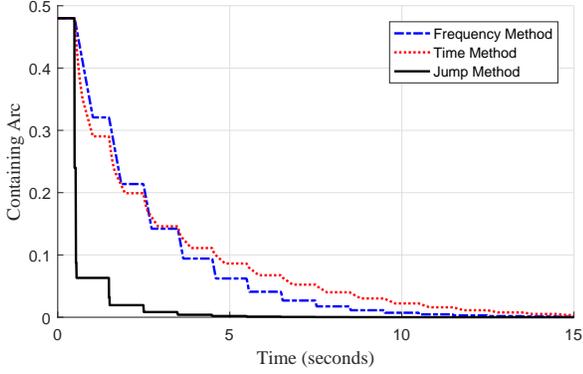}
	\centering
	\caption{Containing arcs, \(\Lambda\) as a function of time for the networks in Fig. \ref{fig:PRCSyncPhaseEvolution2}. The convergence speed of the containing arcs under the constant frequency method, with \(\omega_a = 0.3\omega_0\), and the constant time method, with \(\tau = 0.3\) seconds, is reduced compared with the phase jump case.}
	\label{fig:PRCSyncArcs2}
\end{figure}

\subsection{PRC Synchronization}

We first simulate the PRC synchronization algorithm given in \cite{EnergyEfficientSync2012}. Fig. \ref{fig:PRCSyncPhaseEvolution1} shows that the network does indeed synchronize. As expected, the network converges more slowly when under the phase continuity methods from Sec. \ref{sec:continuity}, due to the reduced effective coupling strengths of the oscillators at most of the firing events. The rate of convergence is illustrated in Fig. \ref{fig:PRCSyncArcs1} where the length of the containing arc is plotted as a function of time.

\begin{figure}[t] 
	\includegraphics[width=\columnwidth]{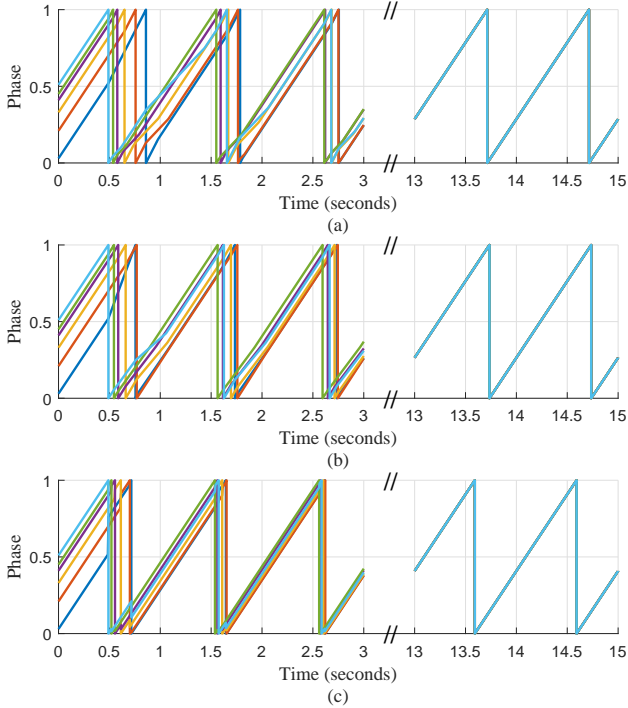}
	\centering
	\caption{Phase Evolution of PRC Synchronization algorithm in \cite{EnergyEfficientSync2012} for \(N = 6\) oscillators in a ring topology, \(\alpha = 0.5\), no refractory period, and random initial conditions in a containing arc \(\Lambda < \frac{1}{2}\). (a) Continuous phase evolution under the constant frequency method, with \(\omega_a = 0.3\omega_0\). (b) Continuous phase evolution under the constant time method, with \(\tau = 0.3\) seconds. (c) Phase jumps, for comparison to the phase continuity methods.}
	\label{fig:PRCSyncPhaseEvolution3}
\end{figure}

\begin{figure}[t] 
	\includegraphics[width=\columnwidth]{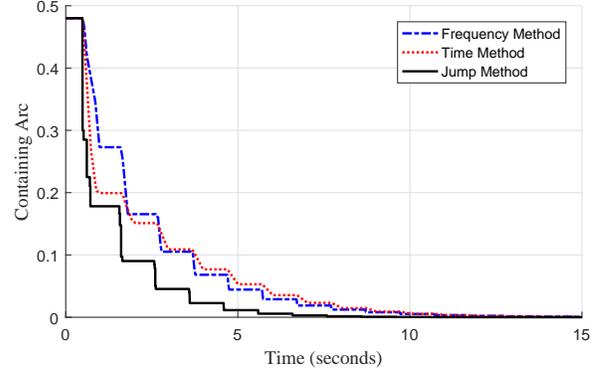}
	\centering
	\caption{Containing arcs, \(\Lambda\) as a function of time for the networks in Fig. \ref{fig:PRCSyncPhaseEvolution3}. The convergence speed of the containing arcs under the constant frequency method, with \(\omega_a = 0.3\omega_0\), and the constant time method, with \(\tau = 0.3\) seconds, is reduced compared with the phase jump case.}
	\label{fig:PRCSyncArcs3}
\end{figure}

As indicated in \cite{EnergyEfficientSync2012}, a refractory period of length \(D\) can be incorporated into the PRC synchronization algorithm, as shown in Fig. \ref{fig:PRCRefractory1}, and synchronization can still be achieved. We illustrate in Fig. \ref{fig:PRCSyncPhaseEvolution2} that the network will synchronize using the phase continuity methods in Sec. \ref{sec:continuity}. As can be seen in Fig. \ref{fig:PRCSyncArcs2}, convergence is slower than when we excluded the refractory period in Fig. \ref{fig:PRCSyncArcs1}. Again, the reduced coupling strength of the network due to phase continuity leads to an increased convergence time.

The PRC synchronization algorithm in \cite{EnergyEfficientSync2012} also achieves convergence in networks with a generally connected topology. We illustrate this case in Fig. \ref{fig:PRCSyncPhaseEvolution3},  where we use the general bidirectional ring topology. Again, we see in Fig. \ref{fig:PRCSyncArcs3} that the use of the phase continuity methods in Sec. \ref{sec:continuity} still allow the network to synchronize, although the convergence rate is decreased due to the reduced effective coupling strength of the oscillators at various firing instances.

\subsection{Peskin Synchronization Algorithm}
We will next consider the original PCO model that was first introduced by Peskin in \cite{peskin:75}. He described the oscillators as ``integrate-and-fire" oscillators, increasing in phase and firing and resetting their phase when they reached a threshold. The phase of an oscillator is mapped onto a state variable, \(x_i(t)\), using the relation \(x_i(t)=f(\theta_i(t))\), where \(f(\theta)\) is a function that is ``smooth, monotonic increasing, and concave down" \cite{mirollo:90}. When a pulse is received, the oscillator maps its current phase to the state variable, increments the state variable by an amount \(\epsilon\), and then maps the state back to the phase using the inverse function \(g(x) = f^{-1}(x)\). That is, the new phase of the oscillator can be written as
\begin{equation} \label{eq:PeskinMap}
\theta_{i}^{+}(t) = g(f(\theta_i(t))+\epsilon)
\end{equation}

If the state variable is incremented past the threshold value (i.e., \(f(\theta_i(t)) + \epsilon > 1\)), then the oscillator immediately fires and resets its phase to zero, and becomes completely synchronized with the oscillator that had fired previously.

Any function \(f(\theta)\) that meets the requirements as above can be used to map the phase into the state variable. Peskin used the following function and its associated inverse:
\begin{equation} \label{eq:Peskinf}
f(\theta) = (1-e^{-\gamma})(1-e^{-\gamma \theta})
\end{equation}
\begin{equation} \label{eq:Pesking}
g(x) = \frac{1}{\gamma} \ln(\frac{1-e^{-\gamma}}{1-e^{-\gamma}-x})
\end{equation}

Mirrolo and Strogatz further improved Peskin's model in \cite{mirollo:90}, and used an alternate function for mapping the phase to the state variable:
\begin{equation} \label{eq:MSf}
f(\theta) = \frac{1}{b} \ln(1+(e^{b}-1)\theta)
\end{equation}
\begin{equation} \label{eq:MSg}
g(x) = \frac{e^{bx}-1}{e^{b}-1}
\end{equation}

Using these functions, an equivalent phase response curve (PRC) can be found by determining \(\phi_i\) from (\ref{eq:PhaseChange}). Fig. \ref{fig:PRC2} illustrates these equivalent PRCs.

\begin{figure}[t] 
	\includegraphics[width=\columnwidth]{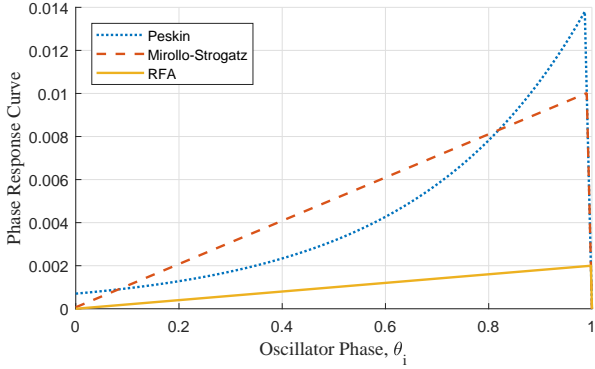}
	\centering
	\caption{Phase Response Curve (PRC) equivalents for the Peskin and Mirollo-Strogatz synchronization algorithms, and the Reachback Firefly Algorithm (RFA).}
	\label{fig:PRC2}
\end{figure}

It is important to note that the Peskin PCO model does not incorporate any kind of coupling strength parameter, \(\alpha\), as in the PRC synchronization given in \cite{EnergyEfficientSync2012}. This is easily verified by noting that the state variable increment \(\epsilon\) does not simply scale the equivalent PRC function, but modifies the overall shape of the function. The Peskin model assumes that the phase jumps the entire amount necessary according to the state mapping function parameters. Equivalently, the Peskin model assumes that the coupling strength \(\alpha\) for the network is always \(1\).

The lack of a coupling strength parameter makes it difficult to extend the analysis of the previous sections to the Peskin model. However, even though there is no coupling strength inherent to the Peskin model, simulations show good synchronization results when the phase continuity methods are applied. For example, using the state variable created from the functions given in (\ref{eq:Peskinf}) and (\ref{eq:Pesking}), and the equivalent PRC function, we find the amount that the oscillator would jump and apply phase continuity methods from section \ref{sec:continuity}. Fig. \ref{fig:PeskinSyncArcs} shows that the phase continuity methods still allow the PCO network to synchronize.

Similarly, using the alternate state variable function introduced by Mirollo and Strogatz given in (\ref{eq:MSf}) and (\ref{eq:MSg}) also allows the network to synchronize under the phase continuity methods, as shown in Fig. \ref{fig:MSSyncArcs}. As expected, the convergence rates are comparable for the phase continuity methods.

\begin{figure}[t] 
	\includegraphics[width=\columnwidth]{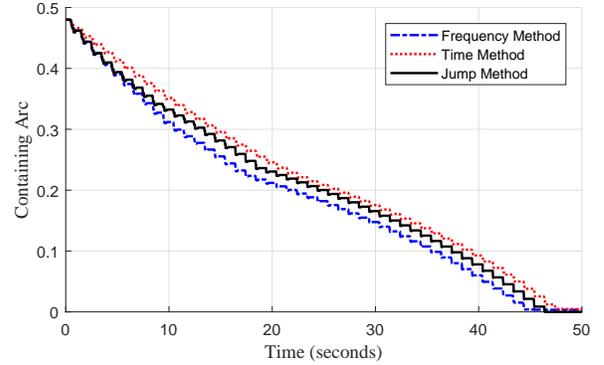}
	\centering
	\caption{Containing arcs, \(\Lambda\), as a function of time for the networks under the Peskin synchronization model with parameters \(\epsilon = 0.002, \gamma = 3\). The convergence property of the containing arcs under the constant frequency method, with \(\omega_a = 0.3\omega_0\), the constant time method, with \(\tau = 0.1\) seconds, is maintained compared to the phase jump case.}
	\label{fig:PeskinSyncArcs}
\end{figure}

\begin{figure}[t] 
	\includegraphics[width=\columnwidth]{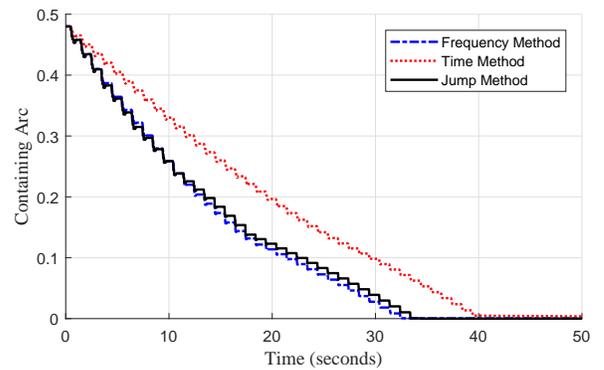}
	\centering
	\caption{Containing arcs, \(\Lambda\), as a function of time for the networks under the Mirollo-Strogatz synchronization model with parameters \(\epsilon = 0.002, b = 5\). The convergence property of the containing arcs under the constant frequency method, with \(\omega_a = 0.3\omega_0\), the constant time method, with \(\tau = 0.1\) seconds, is maintained compared to the phase jump case.}
	\label{fig:MSSyncArcs}
\end{figure}

\subsection{Reachback Firefly Algorithm}
Another synchronization algorithm that has been proposed is the Reachback Firefly Algorithm (RFA) by Werner-Allen et. al. in \cite{Werner-Allen:2005:FSN:1098918.1098934}. This algorithm is based on the Peskin PCO model, where the phase is mapped to a state variable. RFA uses a simple mapping function to decrease computational complexity.
\begin{equation} \label{eq:RFAf}
f(\theta) = \ln(\theta)
\end{equation}
\begin{equation} \label{eq:RFAg}
g(x) = e^{x}
\end{equation}

Fig. \ref{fig:PRC2} also illustrates the equivalent PRC function for the functions used in the RFA model.

The key difference between the RFA model and Peskin algorithm is that the oscillators wait to jump until the moment they fire. As the oscillator receives pulses, it records how much it would jump, according to the state variable mapping, at each time instance. When the oscillator reaches the threshold and fires, it then adds all of the recorded jump amounts from the previous cycle, and jumps by the total amount. This process is then repeated for each cycle until synchronization is achieved.

\begin{figure}[t] 
	\includegraphics[width=\columnwidth]{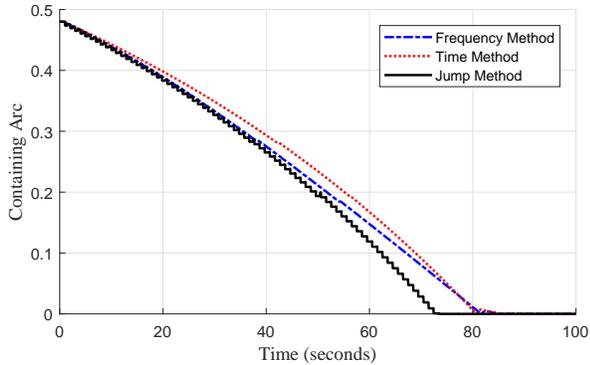}
	\centering
	\caption{Containing arcs, \(\Lambda\), as a function of time for the networks under the Reachback Firefly Algorithm (RFA) with parameter \(\epsilon = 0.002\). The convergence property of the containing arcs under the constant frequency method, with \(\omega_a = 0.007\omega_0\), the constant time method, with \(\tau = 1.1\) seconds, is maintained compared to the phase jump case. The RFA algorithm further allows for phase continuity parameters to be in a broader interval than the Peskin algorithm.}
	\label{fig:RFASyncArcs}
\end{figure}

As with the Peskin model, there is no inherent coupling strength parameter \(\alpha\) in the RFA model. The coupling strength \(\alpha\) is assumed to be always \(1\). This lack of coupling strength incorporated into the RFA model makes it similarly difficult to extend the analysis from the previous sections. However, like with the Peskin model, simulations show good synchronization results when the phase continuity methods are applied. Fig. \ref{fig:RFASyncArcs} illustrates that the phase continuity methods in Sec. \ref{sec:continuity} allow the PCO network to synchronize. It is also important to note that the phase continuity method parameters can be within a broader interval than for the simpler Peskin model. Since each oscillator does not jump when it receives a pulse, and only when it fires, the effective coupling strength can be maximized more easily than in the Peskin model.

\section{Conclusions} \label{sec:conclusion} 
In this paper, we consider the problem of synchronizing clocks while guaranteeing time continuity. To do so, we utilize and analyze the behavior of pulse-coupled oscillator (PCO) networks under the constraint of phase continuity. The original PCO network model used discontinuous phase jumps, but sharp discontinuities in the phase variable of the oscillators is usually not desirable. We present a pair of methods in which the phase variable of an oscillator is able to evolve in a continuous manner.

Using these phase continuity methods, we show that the behavior of the network under these restraints can be modeled as a time-varying coupling strength in the network. Specifically, the coupling strength may be effectively reduced between firing instances. We mathematically prove that a pulse-coupled oscillator network will synchronize with a time-varying coupling strength using a delay-advance phase response curve. Overall, PCO networks under various synchronization algorithms can still achieve desirable behavior using continuous evolutions in the phase variable.

The Peskin model for pulse-coupled oscillators does not include a coupling strength parameter. The model rather assumes that it is constant. Further research may be desirable to consider the effects of having a reduced, and possible time-varying, coupling strength in the Peskin model.

Other phase continuity methods besides the ones proposed in this paper are possible. Specifically, methods that ensure a continuous change in oscillator frequency, rather than just a continuous change in oscillator phase, may be desirable in certain applications. Such methods may require more strict conditions to ensure the desired convergence properties. General analysis of the myriad of other PCO algorithms under phase continuity and time-varying coupling strength reveals much research potential.

\section{Acknowledgment} \label{sec:acknowledge}
We would like to thank those who reviewed the initial drafts of this paper. 


\bibliographystyle{unsrt}
\bibliography{Sync_Clock_Cont_double_column.bib}

\end{document}